\documentclass[submission,copyright,creativecommons]{eptcs}
\providecommand{\event}{TARK 2017} 
\usepackage{breakurl}             
\usepackage{underscore}           

\usepackage{amsmath}
\usepackage{amssymb}
\usepackage{amsthm}

\newtheorem{lemma}{Lemma}
\newtheorem{definition}{Definition}
\newtheorem{corollary}{Corollary}
\newtheorem{proposition}{Proposition}

\renewcommand{\L}{\mathbb{L}}

\renewcommand{\emph}{\textbf}

\newcommand{\Prop}{\mathsf{Prop}}
\newcommand{\Nom}{\mathsf{Nom}}
\newcommand{\Cnom}{\mathsf{Cnom}}

\newcommand{\noma}{\mathbf{a}}

\newcommand{\cnomx}{\mathbf{x}}

\renewcommand{\L}{\mathcal{L}}

\newcommand{\Lmu}{\mathcal{L}_C}
\newcommand{\Ag}{\mathsf{Ag}}

\newcommand{\nomb}{\mathbf{b}}
\newcommand{\nomx}{\mathbf{x}}
\newcommand{\val}[1]{[\![{#1}]\!]}
\newcommand{\descr}[1]{(\![{#1}]\!)}
\renewcommand{\phi}{\varphi}

\title{Toward an Epistemic-Logical Theory of Categorization}

\author{Willem Conradie
\institute{Department of Pure and Applied\\ Mathematics,\\ University of Johannesburg\\
Johannesburg, South Africa}
\email{wconradie@uj.ac.za}
\and
Sabine Frittella 
\institute{INSA Centre Val de Loire,\\
Univ. Orl\' eans, LIFO EA 4022\\
Bourges, France}
\email{sabine.frittella@insa-cvl.fr}
\and
Alessandra Palmigiano
\institute{Faculty of Technology, Policy and\\ Management,\\ Delft University of Technology\\
Delft, the Netherlands}
\institute{Department of Pure and Applied\\ Mathematics,\\ University of Johannesburg\\
Johannesburg, South Africa}
\email{A.Palmigiano@tudelft.nl}
\and
Michele Piazzai \qquad\qquad Apostolos Tzimoulis
\institute{Faculty of Technology, Policy and Management,\\ Delft University of Technology\\
Delft, the Netherlands}
\email{\quad m.piazzai@tudelft.nl \quad\qquad a.tzimoulis-1@tudelft.nl}
\and
Nachoem M.\ Wijnberg
\institute{Amsterdam Business School,\\ University of Amsterdam\\
Amsterdam, the Netherlands}
\institute{Faculty of Economic and Financial\\ Sciences and Faculty of Management,\\ University of Johannesburg\\
Johannesburg, South Africa}
\email{n.m.wijnberg@uva.nl}
}

\def\titlerunning{Toward an Epistemic-Logical Theory of Categorization}
\def\authorrunning{W.\ Conradie, S.\ Frittella, A.\ Palmigiano, M.\ Piazzai, A.\ Tzimoulis \& N.\ Wijnberg}
\def\copyrightholders{Conradie, Frittella,  Palmigiano, Piazzai, Tzimoulis \&  Wijnberg}

\begin{document}
\maketitle

\begin{abstract}
Categorization systems are widely studied in psychology, sociology, and organization theory as information-structuring devices which are critical to decision-making processes. 
In the present paper, we introduce a sound and complete  epistemic  logic of categories and agents' categorical perception. The Kripke-style semantics of this logic is given in terms of data structures based on two domains: one domain representing objects (e.g.~market products) and one domain representing the features of the objects which are relevant to the agents' decision-making.
 We use this framework to discuss and propose logic-based formalizations of some core concepts from psychological, sociological, and organizational research in categorization theory. 
\end{abstract}

\section{Introduction}

Categories (understood as types of collective identities for broad classes of objects or of agents) are the most basic cognitive tools, and are key to the use of language, the construction of knowledge and identity, and the formation of agents' evaluations and decisions. The literature on categorization is expanding rapidly, motivated by--and  in connection with--the theories and methodologies of a wide range of fields in the social sciences and AI. For instance, in linguistics, categories are central to the mechanisms of grammar generation \cite{Cr91};  in AI, classification techniques are core to pattern recognition, data mining, text mining and knowledge discovery in databases;  in sociology, categories are used to explain the construction of social identity \cite{Je00};  in management science, categories are used to predict how products and producers will be perceived and evaluated by consumers and investors \cite{hsu2011typecasting,wijnberg-gap,wijnberg2011classification,PD16,CPT17}.

In \cite{CFPPTW16}, we proposed the  framework of a positive (i.e.~negation-free and implication-free) normal multi-modal logic as an epistemic logic of categories and agents' categorical perception, and discussed its algebraic and Kripke-style semantics. In the present paper, we introduce  a simpler and more general framework than \cite{CFPPTW16}, in which the (rather technical) restrictions on the Kripke-style models of \cite{CFPPTW16} are dropped. We use this logical framework to formalize core notions developed and applied in the fields mentioned above, with a focus on those relevant to management science, as a step towards building systematic connections between modern categorization theory and epistemic logic. 

\paragraph{Structure of the paper.} In Section \ref{sec:foundations formal}, we briefly review the main views on the foundations of categorization theory together with the formal approaches inspired by some of these views. In Section \ref{sec:logics}, we discuss the basic framework of the {\em epistemic logic of categories}. We introduce its refined Kripke-style semantics and axiomatization, together with two language enrichments involving a common knowledge-type construction and hybrid-style nominal (and co-nominal) variables, respectively. In Section \ref{sec:CKFP}, we discuss a number of core categorization-theoretic notions from business science and our proposed formalizations of them. In Section \ref{sec:CCL}, we discuss further directions. In the appendix (Section \ref{sec:Appendix}), 
we discuss the
soundness and completeness of the logics introduced in Section \ref{sec:logics}.

\section{Categorization: foundations and formal approaches}
\label{sec:foundations formal}
In the present section we review the main views, insights, and approaches to the foundations of categorization theory and to the formal models capturing these. Our account is necessarily incomplete. We refer the reader to \cite{CL05} for an exhaustive overview.
\subsection{Extant foundational approaches} The literature on the foundations of categorization theory displays a variety of definitions, theories, models, and methods, each of which capturing some key facets of categorization. The {\em classical} theory of categorization \cite{SM81} goes back to Aristotle, and is based on the insight that all members of a category share some fundamental features which define their membership. Accordingly, categorization is viewed as a deductive process of reasoning with necessary and sufficient conditions, resulting in categories with sharp boundaries, which are represented equally well by any of their members. The classical view has inspired influential approaches in machine learning such as {\em conceptual clustering} \cite{Fi87}.
However, this view runs into difficulties when trying to accommodate a new object or entity which would intuitively be part of a given category but does not share all the defining features of the category. Other difficulties, e.g.~providing an exhaustive list of defining features, unclear cases, and the existence of members of given categories which are judged to be better representatives of the whole class than others, motivated the introduction of {\em prototype} theory \cite{La99,Ro05}.
This theory regards categorization as the inductive process of finding the best match between the features of an object and those of the closest prototype(s).
Prototype theory addresses the above mentioned problems of the classical theory by relaxing the requirement that membership be decided through the satisfaction of an exhaustive list of features. It allows for unclear cases and embraces the empirically verified intuition that people regard membership in most categories as a matter of degrees, with certain members being more central (or prototypical) than others. (For instance, robins are regarded as prototypical birds, while  penguins are not.)
To account for how an ex-ante prototype is generated in the mind of agents, the {\em exemplar} theory \cite{SM02} was proposed, according to which individuals make category judgments by comparing new stimuli with instances already stored in memory (the ``exemplars''). 
However, the existence of instances or prototypes of a given category presupposes that this category has already been defined. Hence, both the prototype and the exemplar view run into a circularity problem. Moreover, it has been argued that {\em similarity-based} theories of ca\-te\-gorization (such as the prototype and the exemplar view) fail to address the problem of explaining `why we have the categories we have', or, in other words, why certain categories seem to be more cogent and coherent than others. 
Even more fundamentally, similarity might be imposed rather than discovered (do things belong in the same category because they are similar, or are they similar because they belong in the same category?), i.e.~similarity might be the effect of conceptual coherence rather than its cause. Pivoting on the notion of coherence for category-formation, the {\em theory-based} view on categorization \cite{MM85} posits that categories arise in connection with theories (broadly understood so as to include also informal explanations). For instance, ice, water and steam can be grouped together in the same category on the basis of the theory of phases in physical chemistry. The coherence of categories proceeds from the coherence of the theories on which they are based. This view of categorization allows one to group together entities which would be scored as dissimilar using different methods; for instance, it allows to group together a gold watch, the school report of one's grandfather, and the 
ownership of a piece of land in the category of ``things one wants one's children to inherit'', which is based on one's theory of what family is. However, the theory-based view does not account for the intuition that categories themselves are the building blocks of theory-formation, which again results in a circularity problem. Summing up, the extant views on categorization (the classical \cite{SM81}, prototype \cite{La99,Ro05}, exemplar \cite{SM02}, and theory-based \cite{MM85}) are difficult to reconcile and merge into a satisfactory overarching theory accommodating all the insights into categorization that researchers in the different fields have been separately developing.  The present paper is one of the first steps of a research program aimed at clarifying notions developed independently, and at developing a common ground which can hopefully facilitate the build-up of such a theory. 

\subsection{Extant formal approaches}
\label{ssec:formal}
\paragraph{Conceptual spaces.} The formal approach to the representation of categories and concepts which is perhaps the most widely adopted in social science and management science is  the one  introduced by G\"{a}rdenfors, which is based on  {\em conceptual spaces} \cite{Ga04}. These are multi-dimensional geometric structures, the components of which (the quality dimensions) are intended to represent basic features -- e.g.~colour, pitch, temperature, weight, time, price -- by which objects (represented as points in the product space of these dimensions) can be meaningfully compared. Each dimension is endowed with its appropriate geometric (e.g.~metric, topological) structure. Concept-formation in conceptual spaces is modelled according to a {\em similarity-based} view of concepts. Specifically, if each dimension of a conceptual space has a metric, these metrics translate in a notion of distance between the objects represented in the space, which models their similarity, so that the closer their distance, the more similar they are. {\em Concepts} (i.e.~formal categories) are represented as {\em convex} sets of the conceptual space\footnote{A subset is {\em convex} if it includes the segments between any two of its points. In the Euclidian plane, squares are convex while stars are not.}. The geometric center of any such concept is a natural interpretation of the prototype of that concept. 

\paragraph{Formal Concept Analysis.} A very different approach, Formal Concept Analysis (FCA) \cite{Wille}, is a method of data analysis based on Birkhoff's representation theory of complete lattices \cite{DaveyPriestley2002}. In FCA, databases are represented as {\em formal contexts}, i.e.~structures $\left( A, X, I \right)$ such that $A$ and $X$ are sets, and $I\subseteq A\times X$ is a binary relation. Intuitively, $A$ is understood as a collection of {\em objects}, $X$ as a collection of {\em features}, and for any object $a$ and feature $x$, the tuple $\left(a, x\right)$ belongs to $I$ exactly when object $a$ has feature $x$.
Every formal context $\left(A, X, I \right)$ can be associated with the collection of its {\em formal concepts}, i.e.~the tuples $\left(B, Y \right)$ such that $B\subseteq A$, $Y\subseteq X$, and $B\times Y$ is a maximal rectangle included in $I$. The set $B$ is the {\em extent} of the formal concept $\left( B, Y \right)$, and $Y$ is its {\em intent}.
Because of maximality, the extent of a formal concept uniquely identifies and is identified by its intent.  Formal concepts can be partially ordered; namely, $\left(B, Y\right)$ is a {\em subconcept} of $\left(C, Z\right)$ exactly when $B\subseteq C$, or equivalently, when $Z\subseteq Y$. Ordered in this way, the concepts of a formal context form a {\em complete lattice} (i.e., the least upper bound and the greatest lower bound of every collection of formal concepts exist), and by Birkhoff's theorem, every complete lattice is isomorphic to some concept lattice.
The link established by FCA between complete lattices and the formalization of concepts (or categories) captures an aspect of categories which is very much highlighted in the categorization theory literature. Namely,  categories never occur in isolation; rather, they arise in the context of {\em categorization systems} (e.g.\ taxonomies), which are typically organized in hierarchies of super- (i.e.\ less specified) and sub- (i.e.\ more specified) categories.
 While most approaches identify concepts with their extent, in FCA, intent and extent of a concept are treated on a par,
 i.e., the intent of a concept is just as essential as its extent.  While FCA has tried to connect itself with various cognitive and philosophical theories of concept-formation, it is most akin to the {\em classical} view.

\paragraph{Formal concepts as modal models.} In \cite{CFPPTW16}, we first established a connection between FCA and modal logic, based on the idea  that (enriched) formal contexts can be taken as models of an epistemic modal logic of categories/concepts. Formulas of this logic are constructed out of a set of atomic variables using the standard positive propositional connectives $\wedge, \vee, \top, \bot$, and modal operators $\Box_i$ associated with each agent $i\in \mathsf{Ag}$. The formulas so generated do not denote states of affairs (to which a truth-value can be assigned), but {\em categories} or {\em concepts}.
In this modal  language, as usual, it is easy to distinguish  the `objective' or factual information (stored in the database), encoded in the formulas of the modal-free fragment of the language, and  the agents' subjective interpretation of the `objective' information, encoded in formulas in which modal operators occur.
In this language, we can talk about e.g.~the category that according to Alice is the category that according to  Bob is the category of Western movies.
This makes it possible to define {\em fixed points} of these regressions, similarly to the way in which common knowledge is defined in classical epistemic logic \cite{fagin2003reasoning}.
Intuitively, these fixed points represent the stabilization of a process of social interaction; for instance, the consensus reached by a group of agents regarding a given category.

Models for this  logic are formal contexts $\left(A, X, I\right)$ enriched with an extra relation $R_i\subseteq A\times X$ for each agent (intuitively, for every object $a\in A$ and every feature $x\in X$, we read $aR_ix$ as `object $a$ has feature $x$ according to agent $i$'. Hence, while the relation $I$ represents reality as is recorded in the database represented by the formal context $\left(A, X, I\right)$, each relation $R_i$ represents as usual the subjective view of the corresponding agent $i$ about objects and their features, and is used to interpret $\Box_i$-formulas. 

This logic arises and has been studied in the context of  unified correspondence theory \cite{CoGhPa14}, and allows one to  relate, via Sahlqvist-type results, sentences in the first-order language of enriched formal contexts (expressing low-level, concrete conditions about objects and features) with inequalities $\phi\leq \psi$, where $\phi$ and $\psi$ are formulas in the modal language above, expressing high-level, abstract relations about categories and how they are perceived and understood by different agents.
In the next section, we expand on the relevant definitions and background facts about this logic.

\section{Epistemic logic of categories}
\label{sec:logics}
\paragraph{Basic logic and intended meaning.} Let $\Prop$ be a (countable or finite) set of atomic propositions and $\Ag$ be a finite set (of agents). The basic language $\L$ of the epistemic logic of categories is
\[ \varphi := \bot \mid \top \mid p \mid  \varphi \wedge \varphi \mid \varphi \vee \varphi \mid \Box_i \varphi, \]
where $p\in \Prop$. As mentioned above, formulas in this language are terms denoting categories (or concepts). Atomic propositions provide a vocabulary of {\em category labels}, such as music genres (e.g.~jazz, rock, rap), movie genres (e.g.~western, drama, horror), supermarket products (e.g.~milk, dairy products, fresh herbs). Compound formulas $\phi\wedge \phi$ and $\phi\vee\psi$ respectively denote the greatest common subcategory and the smallest  common supercategory of $\phi$ and $\psi$. For a given agent $i\in \Ag$, the formula $\Box_i\phi$ denotes the category  $\phi$, according to $i$. At this stage we are deliberately vague as to the precise meaning of `according to'. Depending on the properties of $\Box_i$, the formula $\Box_i\phi$ might denote the category {\em known}, or {\em perceived}, or {\em believed} to be $\phi$ by agent $i$.
		The {\em basic}, or {\em minimal normal} $\mathcal{L}$-{\em logic} is a set $\mathbf{L}$ of sequents $\phi\vdash\psi$ (which intuitively read ``$\phi$ is a subcategory of $\psi$'') with $\phi,\psi\in\mathcal{L}$, containing the following axioms:
		\begin{itemize}
			\item Sequents for propositional connectives:
			\begin{align*}
				&p\vdash p, && \bot\vdash p, && p\vdash \top, & &  &\\
				&p\vdash p\vee q, && q\vdash p\vee q, && p\wedge q\vdash p, && p\wedge q\vdash q, &
			\end{align*}
			\item Sequents for modal operators:
			\begin{align*}
				&\top\vdash \Box_i \top &&
                \Box_i p\wedge \Box_i  q \vdash \Box_i \left( p\wedge q\right)
			\end{align*}
		\end{itemize}
		and closed under the following inference rules:
		\begin{displaymath}
			\frac{\phi\vdash \chi\quad \chi\vdash \psi}{\phi\vdash \psi}
			\quad\quad
			\frac{\phi\vdash \psi}{\phi\left(\chi/p\right)\vdash\psi\left(\chi/p\right)}
			\quad\quad
			\frac{\chi\vdash\phi\quad \chi\vdash\psi}{\chi\vdash \phi\wedge\psi}
			\quad\quad
			\frac{\phi\vdash\chi\quad \psi\vdash\chi}{\phi\vee\psi\vdash\chi}
			\quad\quad
			\frac{\phi\vdash\psi}{\Box_i \phi\vdash \Box_i \psi}
		\end{displaymath}
Thus, the modal fragment of $\mathbf{L}$ incorporates the viewpoints of individual agents into the syllogistic reasoning supported by the propositional fragment of $\mathbf{L}$. By an {\em $\mathcal{L}$-logic}, we understand any  extension of $\mathbf{L}$  with $\mathcal{L}$-axioms $\phi\vdash\psi$.
	
\paragraph{Interpretation in enriched formal contexts.} Let us discuss the structures which play the role of Kripke frames.\footnote{Details can be found in Section \ref{sec:Appendix}.}  An {\em enriched formal context} is a tuple
$$\mathbb{F} = \left(\mathbb{P}, \{R_i\mid i\in \Ag\}\right)$$
such that $\mathbb{P} = \left(A, X, I\right)$ is a formal context, and $R_i\subseteq A\times X$ for every $i\in \Ag$, satisfying certain additional properties which guarantee that their associated modal operators are well defined  (cf.~Definition \ref{def:compositional efc}). As mentioned above, formal contexts represent databases of market products (the elements of the set $A$), relevant features (the elements of the set $X$), and an incidence relation $I\subseteq A\times X$ (so that $aIx$ reads: ``market product $a$ has feature $x$''). In addition, enriched formal contexts contain information about the epistemic attitudes of individual agents, so that $aR_ix$ reads: ``market product $a$ has feature $x$ according to agent $i$'', for any $i\in \Ag$. A {\em valuation} on $\mathbb{F}$ is a map $V:\Prop\to \mathcal{P}\left(A\right)\times \mathcal{P}\left(B\right)$, with the restriction that $V\left(p\right)$ is a {\em formal concept} of $\mathbb{P} = \left(A, X, I\right)$, i.e., every $p\in \Prop$ is mapped to $V\left(p\right) = \left(B, Y\right)$ such that $B\subseteq A$, $Y\subseteq X$, and $B\times Y$ is  maximal rectangle contained in $I$. For example, if $p$ is the category-label denoting western movies, and $\mathbb{P}$ is a given database of  movies (stored in $A$) and movie-features (stored in $X$), then  $V$ interprets the category-label $p$  in the model $\mathbb{M} = \left(\mathbb{F}, V\right)$ as the formal concept (i.e.~semantic category) $V\left(p\right) = \left(B, Y\right)$,  specified by  the set of movies $B$  (i.e.~the set of western movies of the database) and  by the set of movie-features $Y$ (i.e.~the set of features which all western movies have). The elements of $B$ are the {\em members} of category $p$ in  $\mathbb{M}$; the elements of $Y$ {\em describe}  category $p$ in $\mathbb{M}$. The set $B$ (resp.~$Y$) is the {\em extension} (resp.~the {\em description}) of $p$ in $\mathbb{M}$, and sometimes we will denote it $\val{p}_{\mathbb{M}}$ (resp.~$\descr{p}_{\mathbb{M}}$) or $\val{p}$ (resp.~$\descr{p}$) when it does not cause confusion. Alternatively, we write:
\begin{center}
\begin{tabular}{llll}
$\mathbb{M}, a \Vdash p$ & iff & $a\in \val{p}_{\mathbb{M}}$  & \\
$\mathbb{M}, x \succ p$ & iff & $x\in \descr{p}_{\mathbb{M}}$  &\\
\end{tabular}
\end{center}
and we read $\mathbb{M}, a \Vdash p$ as ``$a$ is a member of $p$'', and $\mathbb{M}, x \succ p$ as ``$x$ describes $p$''.
The interpretation of atomic propositions can be extended to propositional $\mathcal{L}$-formulas as follows:
\begin{center}
\begin{tabular}{llll}
$\mathbb{M}, a \Vdash\top$ &  & always \\
$\mathbb{M}, x \succ \top$ & iff &   $a I x$ for all $a\in A$\\
$\mathbb{M}, x \succ  \bot$ &  & always \\
$\mathbb{M}, a \Vdash \bot $ & iff & $a I x$ for all $x\in X$\\
$\mathbb{M}, a \Vdash \phi\wedge \psi$ & iff & $\mathbb{M}, a \Vdash \phi$ and $\mathbb{M}, a \Vdash  \psi$ & \\
$\mathbb{M}, x \succ \phi\wedge \psi$ & iff & for all $a\in A$, if $\mathbb{M}, a \Vdash \phi\wedge \psi$, then $a I x$\\
$\mathbb{M}, x \succ \phi\vee \psi$ & iff &  $\mathbb{M}, x \succ \phi$ and $\mathbb{M}, x \succ  \psi$ &\\
$\mathbb{M}, a \Vdash \phi\vee \psi$ & iff & for all $x\in X$, if $\mathbb{M}, x \succ \phi\vee \psi$, then $a I x$  & \\
\end{tabular}
\end{center}
Hence, in each model, $\top$ is interpreted as the category generated by the set $A$ of all objects, i.e.~the widest category and hence the one with the laxest (possibly empty) description;  $\bot$ is interpreted  as the category generated by the set $X$ of all features, i.e.~the smallest (possibly empty) category and hence the one with the most restrictive description; $\phi\wedge\psi$ is interpreted  as the semantic category generated by the intersection of the extensions of $\phi$ and $\psi$ (hence, the description of $\phi\wedge\psi$ certainly includes $\descr{\phi}\cup\descr{\psi}$ but is possibly larger). Likewise,  $\phi\vee\psi$ is interpreted  as the semantic category generated by the intersection of the intensions of $\phi$ and $\psi$ (hence,  objects in $\val{\phi}\cup\val{\psi}$ are certainly members of $\phi\vee\psi$  but there might be others).
As to the interpretation of modal formulas:
\begin{center}
\begin{tabular}{llll}
$\mathbb{M}, a \Vdash \Box_i\phi$ & iff & for all $x\in X$, if $\mathbb{M}, x \succ \phi$, then $a R_i x$& \\
$\mathbb{M}, x \succ \Box_i\phi$ & iff & for all $a\in A$, if $\mathbb{M}, a \Vdash \Box\phi$, then $a I x$.\\
%
\end{tabular}
\end{center}
Thus, in each model,  $\Box_i\phi$ is interpreted as the category whose members are those objects to which {\em agent $i$ attributes}  every feature in the description of $\phi$. Finally, as to the interpretation of sequents:
\begin{center}
\begin{tabular}{llll}
$\mathbb{M}\models \phi\vdash \psi$ & iff & for all $a \in A$, $\mbox{if } \mathbb{M}, a \Vdash \phi, \mbox{ then } \mathbb{M}, a \Vdash \psi$.& 
\end{tabular}
\end{center}
\paragraph{Adding `common knowledge'.} In \cite{CFPPTW16}, we observed that the environment described above is naturally suited to capture not only the factual information and the epistemic attitudes of individual agents, but also the outcome of social interaction. To this effect, we introduce an expansion $\Lmu$ of $\mathcal{L}$  with a common knowledge-type  operator $C$. Given $\Prop$ and $\Ag$ as above,
the  language $\Lmu$ of the epistemic logic of categories with `common knowledge' is:
\[ \varphi := \bot \mid \top \mid p \mid  \varphi \wedge \varphi \mid \varphi \vee \varphi \mid \Box_i \varphi\mid C\left(\varphi\right). \]
 $C$-formulas are interpreted in models as follows:
\begin{center}
\begin{tabular}{llll}
$\mathbb{M}, a \Vdash C\left(\varphi\right)$ & iff & for all $x\in X$, if $\mathbb{M}, x \succ \varphi$, then $a R_C x$& \\
$\mathbb{M}, x \succ C\left(\varphi\right)$ & iff & for all $a\in A$, if $\mathbb{M}, a \Vdash C\left(\varphi\right)$, then $a I x$,\\
%
\end{tabular}
\end{center}
where $R_C\subseteq A\times X$ is defined as $R_C = \bigcap_{s\in S} R_s$, and $R_s\subseteq A\times X$ is the relation associated with the modal operator $\Box_s: = \Box_{i_1}\cdots \Box_{i_n}$ for any element $s= i_1\cdots i_n$ in the set $S$ of finite sequences  of elements of $\Ag$ (cf.~Section \ref{ssec:semantics common knowledge}).

The basic logic of categories with `common knowledge'   is a set $\mathbf{L}_C$ of sequents $\phi\vdash\psi$, with $\phi,\psi\in\Lmu$, which contains the  axioms and is closed under the rules of $\mathbf{L}$, and in addition contains the following axioms:
			\begin{align*}
				& \top\vdash C \left(\top\right) &&
                 C\left( p\right)\wedge C\left(q\right) \vdash C\left( p\wedge q\right) && C\left(p\right)\vdash \bigwedge\{\Box_ip\land\Box_iC\left(p\right)\mid i\in \Ag\}
			\end{align*}
		and is closed under the following inference rules:
		\begin{displaymath}
			\frac{\phi\vdash\psi}{C\left(\phi\right)\vdash C\left(\psi\right)}\quad\quad \frac{\chi\vdash \bigwedge_{i\in \Ag}\Box_i\varphi\quad \{\chi\vdash \Box_i\chi\mid i\in \Ag\}}{\chi\vdash C\left(\varphi\right)}
		\end{displaymath}

\paragraph{Hybrid expansions of the basic language.} In several settings, it is  useful to be able to talk about given objects (market-products) or given features. To this purpose, the languages $\mathcal{L}$ or $\Lmu$ can be further enriched with dedicated sets of variables in the style of hybrid logic.
Let $\Prop$ be a (countable or finite) set of atomic propositions and $\Ag$ be a finite set (of agents). Given $\Prop$ and $\Ag$ as above, and (countable or finite) sets $\Nom$ and $\Cnom$ (of {\em nominals} and {\em conominals} respectively),
the  language $\mathcal{L}_H$ of the hybrid logic of categories is:
\[ \varphi := \bot \mid \top \mid p \mid \noma \mid \nomx \mid  \varphi \wedge \varphi \mid \varphi \vee \varphi \mid \Box_i \varphi, \]
where $i\in \Ag$, $p\in\Prop$, $\noma\in\Nom$ and $\cnomx\in\Cnom$.
A {\em hybrid valuation} on an enriched formal concept $\mathbb{F}$ maps atomic propositions to formal concepts, nominal variables to the formal concepts generated by  single elements of the object domain $A$, and  conominal variables to formal concepts generated by  single elements of the feature domain $X$.
If $V\left(\noma\right)$ is the semantic category generated by $a\in A$, and $V\left(\nomx\right)$ is the semantic category generated by $x\in X$, then nominal and co-nominal variables are interpreted  as follows:
\begin{center}
\begin{tabular}{llll}
$\mathbb{M}, y \succ \noma$ & iff &  $a I y$,\\
$\mathbb{M}, b \Vdash \noma$ & iff &   for all $y\in X$, if $aI y$ then $b I y$& \\
$\mathbb{M}, b \Vdash \nomx$ & iff & $b I x$   \\
$\mathbb{M}, y \succ \nomx$ & iff & for all $b\in A$, if $bI x$ then $b I y$.& \\

%
\end{tabular}
\end{center}

\section{Core concepts and proposed formalizations}
\label{sec:CKFP}
In the present section,  we use  the languages $\mathcal{L}$, $\mathcal{L}_H$ and $\Lmu$ discussed in the previous section to capture  some core notions and properties about categories, appearing and used in the literature in management science, which we discuss in the next subsection. 

\subsection{Core concepts}
\label{ssec:cat:th}
A core issue in management science is how to predict the success of a new market-product, or of a given firm over its competitors. Success clearly depends  on whether the agents in the relevant audiences  decide to buy the product or become clients of the firm, and a key factor in this decision is how each agent resolves a {\em categorization} problem. The  ease with which products or firms are categorized  affects in itself the decision-making, because the more difficult it is to categorize a product or a firm, the higher the cognitive burden and the perceived risk of the decision. 
 This is why 
research has focused on the performances of {\em category-spanning} products or firms (i.e.~products or firms which are members of more than one category).
  While being a member of more than one category can increase visibility and awareness, because audiences interested in any of these categories may pay attention to something which is also in that category, it usually lowers the success.
However, the actual effects of spanning categories will depend on the properties of the categories that are spanned.  The core concepts of categorization theory denote characteristics of categories or of the relation between categories that can be understood to decrease or increase the effects of spanning categories with these particular characteristics.

\paragraph{Typicality.} The issue of whether an object $a$ is a typical member of a given category $\phi$,  or to which extent $a$ is typical of $\phi$, is core to  the {\em similarity-based} views of category-formation \cite{La99,Ro05,SM02}. As mentioned in Section \ref{ssec:formal}, in conceptual spaces, the {\em prototype} of a formal concept is defined as the geometric center of that concept, so that the closer (i.e.~more similar) any other object is to the prototype, the stronger its typicality. While this formalization is visually very appealing, it does not shed much light on the role of the agents in establishing the typicality of an object relative to a category. 

\paragraph{Distance.} The distance between two categories can be defined in different ways. One approach \cite{KH15} is to express it as a negative exponential function of the categories' similarity, where the categories' similarity is calculated using a Jaccard index, i.e., cardinality of the intersection over cardinality of the union. Another approach \cite{PH14}  is to take the Hausdorff distance between the sets in feature space that correspond to the categories. The Hausdorff distance is the maximum of the two minimal point-to-set distances.

\paragraph{Contrast.}
Contrast is defined as the extent to which a category stands out from other categories in the same domain. It is a function of the mean typicality of objects in the category. In a high-contrast category, objects tend to be either very typical members of the category or not members at all \cite{Ha10}. Objects in a high-contrast category tend to be more recognizable to agents and more positively valued \cite{NHR10}. Category spanning leads to greater penalties if the spanned categories have higher contrast \cite{KH10}.

\paragraph{Leniency.} By definition of contrast, members of a low-contrast category $\varphi$ have on average low typicality in that category. This situation is compatible with each of the following alternatives: (a) there are many categories which (according to agents) have members in common with $\varphi$,  (b) there are not many categories which (according to agents) have members in common with $\varphi$. The notion of leniency clarifies this issue.   The leniency of $\varphi$ is defined as the extent to which the members of $\varphi$ are (recognized as) only members of $\varphi$ (and of the other logically unavoidable categories), and not of  other categories \cite{Po12}. 

\subsection{Formalizations}
\label{ssec:formal:cat}
The following proposals are not equivalent to the definitions discussed in the previous subsection, but try to capture their purely qualitative content.
\paragraph{Typicality.} The interpretation of $C$-formulas on models indicates that, for every category $\varphi$, the members of $C(\varphi)$ are those objects which are members of $\varphi$ according to every agent, and moreover, according to every agent, are attributed membership in $\varphi$ by every (other) agent, and so on. This provides justification for our proposal to regard the members of $C(\varphi)$ as the {\em (proto)typical members} of $\varphi$. The main feature of this proposal is that it is explicitly based on the agents' viewpoints. This feature is compatible with empirical methodologies adopted to establish graded membership (cf.~\cite{H2006}). Notice that there is a hierarchy of reasons why a given object fails to be a typical member of $\varphi$, the most severe being that some agents do not recognize its membership in $\varphi$, followed by some agents not recognizing that any other agent would recognize it as a member of $\phi$, and so on. This observation provides a purely qualitative route to encode the {\em gradedness} of (the recognition of) category-membership. That is, two non-typical objects\footnote{represented in the language $\mathcal{L}_H$ as  nominal variables.}  $\noma$ and $\nomb$  can be compared in terms of the minimum number of `epistemic iterations' needed for their typicality test to fail, so that $\nomb$ is {\em more atypical} than $\noma$ if fewer rounds are needed for $\nomb$ than for $\noma$. This definition can be readily adapted so as to say that  $\nomb$ is a {\em more atypical} member of $\psi$ than $\noma$ is of $\varphi$. 

\paragraph{Distance.} For four categories $\varphi, \psi, \chi, \xi$, we can say that $\varphi$ is
{\em closer} to $\psi$  than  $\chi$ is to $\xi$ by means of the sequent $\varphi\vee \psi\vdash \xi\vee \chi$,
the sequent $\xi\wedge \chi\vdash \varphi\wedge \psi$,   or by requiring the two sequents
to hold simultaneously. The first sequent says that $\varphi$ and $\psi$  have
more features in common than $\xi $ and $\chi$ have; the second sequent
says that $\varphi$ and $\psi$ have more common members than $\xi $ and $\chi$ have. Notice that neither the first sequent implies or is implied by the
second. This is why it might be useful to consider the information encoded in both sequents. When  instantiated to  $\varphi = \xi$, these conditions can be used to express that $\varphi$ is closer to $\psi$  than to $\chi$.

\paragraph{Contrast.} If  $\varphi\vdash C(\varphi)$ holds for a category $\varphi$,  every member of $\varphi$ is a typical member of $\varphi$, in the sense discussed above, and hence $\varphi$ has {\em maximal contrast}. Using the formalizations of typicality and distance discussed above, we say that $\varphi$ has {\em equal or higher contrast} than $\psi$ if $\varphi$ is closer to $C(\varphi)$ than $\psi$ is to $C(\psi)$.\footnote{That is, by either requiring that $\varphi\vee C(\varphi)\vdash \psi\vee C(\psi)$, or by requiring that $\psi\vee C(\psi) \vdash \varphi\vee C(\varphi)$, or by requiring both sequents to hold.} 
\paragraph{Leniency.} A category $\varphi$ has {\em no leniency} if its members do not simultaneously belong to other categories. This property can be captured by the following condition: for any $\psi$ and $\chi$, if $\psi\vdash \phi$ and $\psi\vdash \chi$, then either  $\phi\vdash \chi$ or $\chi\vdash \phi$. To understand this condition, let us instantiate $\psi$ as the nominal category $\noma$ (the category generated by one object).
Then $\noma\vdash\varphi$ expresses that the generator of $\noma$  is a member of $\varphi$.   The no-leniency of $\varphi$ would require the generator $a$ of $\noma$ to not belong to other categories. However, the nature of the present formalization constrains $a$ to be a member of every $\chi$ such that $\varphi\vdash \chi$, so $a$ must belong to these categories at least. Also, all the categories $\chi$ such that $\noma\vdash \chi\vdash \phi$ cannot be excluded either, since the possibility that `in-between' categories exist does not depend purely on $a$ and $\varphi$ alone, but depends on the context of other objects and features. Hence, we can understand no-leniency as the requirement that no other categories have $a$ as a member than those of this minimal set of categories which cannot be excluded.

For two categories $\varphi$ and $\psi$, we say that $\varphi$ has {\em greater or equal leniency} than $\psi$ if, for every nominal $\noma$, if $\noma\vdash\psi$  and $\noma\vdash \chi$ for some $\chi$ such that $\chi\nvdash \psi$ and $\psi\nvdash \chi$, then   $\noma\vdash \varphi$ and moreover, $\noma\vdash \xi$ for some category $\xi$ such that $\xi\nvdash \varphi$ and $\varphi\nvdash \xi$.  
Variants of these conditions can be given also in terms of the features (using conominal variables), and also in  terms of the modal operators.

\section{Conclusions and further directions}
\label{sec:CCL}
In this paper, we have introduced a basic epistemic logic of categories, expanded it with `common knowledge'-type and `hybrid logic'-type constructs, 
and used the resulting framework to capture core notions in categorization theory, as developed in management science. The logical formalizations proposed in Section \ref{ssec:formal:cat}  try to capture the purely qualitative content of the original definitions. The essential features of this logical framework make it particularly suitable to emphasize  the different perspectives of individual agents, and how these perspectives interact. The propositional base of these logics is the positive (i.e.~negation-free and implication-free) fragment of classical propositional logic (without distributivity laws). The Kripke-style semantics of this logic is given by structures known as {\em formal contexts} in Formal Concept Analysis \cite{Wille}, which we have {\em enriched} with binary relations to account for the (epistemic) interpretation of the modal operators.  One fundamental difference between this semantics and the classical Kripke semantics for epistemic logics is that the relations directly encode the actual viewpoint of the individual agents, and not their uncertainty or ignorance ($aR_ix$ reads `object $a$ has feature $x$ according to agent $i$).

This paper is still very much a first step, but it already shows how logic can contribute to the vast interdisciplinary area of categorization theory, especially with regard to the analysis of various types of social interaction (e.g.~epistemic, dynamic, strategic). Interestingly, the prospective contributions involve {\em both} technical aspects (some of which we discuss below) {\em and} conceptual aspects (since, as discussed in Section \ref{sec:foundations formal}, there is no single foundational theory or view which exhaustively accounts for all the relevant aspects of categorization).
\paragraph{From RS-frames to arbitrary contexts.}   The  present paper refines previous work \cite{CFPPTW16}, which provides a conceptually independent explanation of the (rather technical) definition of the interpretation clauses of $\mathcal{L}$-formulas on certain enriched formal contexts. These clauses were obtainable as the outcome of mechanical computations (cf.~\cite[Section 2.1.1]{CoPa:non-dist}, \cite[Section A]{CFPPTW16}) the soundness of which was guaranteed by certain facts pertaining to the duality for perfect lattices (cf.~\cite{DGP05,Gehrke}). The treatment in Section \ref{sec:logics} adapts these interpretation clauses to the more general and intuitively more natural category of arbitrary (enriched) formal contexts and their morphisms \cite{Moshier}. 

\paragraph{Fixed points.} One of the most interesting aspects of the present proposal is  that typicality has been  captured with  a `common knowledge' operator. This operator is semantically equivalent to the usual greatest fixed point construction (cf.~Section \ref{sec:Appendix}). This paves the way to  the use of languages expanded with fixed point operators to capture: for instance, as discussed in \cite[Example 4]{CCPZ}, the formula $\nu X.\Box_i(X\wedge p)$ denotes the category obtained as
the limit of a process of ``introspection'' (in which the agent reflects on her perception of
a given category $p$, and on her perception of her perception, and so on). A systematic exploration of this direction is work in progress. 


\paragraph{Proof calculi.} The present framework makes it possible to blend together syllogistic and epistemic reasoning. To further explore those aspects connected with reasoning and deduction in $\mathbf{L}$ and $\mathbf{L}_C$, specifically designed proof calculi will be needed. These calculi will be useful tools to explore the computational properties of these logics; moreover, the conclusions  of formal inferences  can  provide the basis for the development of testable hypotheses. A proof-theoretic account of the basic logic $\mathbf{L}$ can be readily achieved by augmenting the calculus developed in \cite{GP:lattice} for the propositional base  with suitable rules for the modal operators, so as to fall into the general theory of \cite{TrendsXIII}. However, the proof theory of $\mathbf{L}_C$ needs to be investigated. The omega rules introduced in  \cite{FGKP16} might provide a template.

\paragraph{Dynamic epistemic logic of categories.}
An adequate formal account of the dynamic nature of categories is a core challenge facing modern categorization theory. Categories are cognitive tools that agents use as long as they are useful, which is why some categories have existed for millennia and others quickly fade away. Categories shape and are shaped by social interaction. This bidirectional causality is essential to what categories are and do, and this is why the most important and challenging further direction concerns how categories impact on social interaction and how social interaction changes agents' categorizations. One natural step in this direction is to expand the present framework with dynamic modalities, and extend the construction of dynamic updates to models based on enriched formal contexts, as done e.g.~in  \cite{MPS14,KP13}.

\bibliographystyle{eptcs}
\bibliography{aiml16}

\appendix
\section{Soundness and completeness}
\label{sec:Appendix}
In Section \ref{ssec:compatible}, we define $I$-compatible relations and give their properties. In Section \ref{ssec:semantics common knowledge}, we prove that composition of $I$ compatible relations is associative and that the interpretation of $C$ is well-defined. In Section \ref{ssec:soundness}, we prove the soundness of the axioms given in Section \ref{sec:logics}.
In Section \ref{ssec:completeness}, we prove the week completeness of the logics $\mathbb{L}$ and $\mathbb{L}_C$ defined in  Section \ref{sec:logics}.

\subsection{$I$-compatible relations}
\label{ssec:compatible}
In what follows, we fix two sets $A$ and $X$, and use $a, b$ (resp.~$x, y$) for elements of $A$ (resp.~$X$), and $B, C, A_j$ (resp.~$Y, W, X_j$) for subsets of $A$ (resp.~of $X$) throughout this section. For any relation $S\subseteq A\times X$, let 
\[S^\uparrow[B]:=\{x\mid \forall a(a\in B\Rightarrow aSx) \}  \quad\quad S^\downarrow[Y]:=\{a\mid \forall x(x\in Y\Rightarrow aSx) \}.\]
Well known properties of this construction (cf.~\cite[Sections 7.22-7.29]{DaveyPriestley2002}) are stated in the following lemma.
 \begin{lemma}\label{lemma: basic}
\begin{enumerate}
\item $B\subseteq C$ implies $S^\uparrow[C]\subseteq S^\uparrow[B]$, and $Y\subseteq W$ implies $S^\downarrow[W]\subseteq S^\downarrow[Y]$.
 \item $B\subseteq S^\downarrow[S^\uparrow[B]]$ and $Y\subseteq S^\uparrow[S^\downarrow[Y]]$.
 \item $S^\uparrow[B] = S^\uparrow[S^\downarrow[S^\uparrow[B]]]$ and $S^\downarrow[Y] = S^\downarrow[S^\uparrow[S^\downarrow[Y]]]$.
 \item $S^{\downarrow}[\bigcup\mathcal{Y}] = \bigcap_{Y\in \mathcal{Y}}S^{\downarrow}[Y]$ and $S^{\uparrow}[\bigcup\mathcal{B}] = \bigcap_{B\in \mathcal{B}}S^{\uparrow}[B]$.
\end{enumerate}
 \end{lemma}
For any formal context $\mathbb{P}=(A,X,I)$, we sometimes use $B^{\uparrow}$ for $I^{\uparrow}[B]$, and $Y^{\downarrow}$ for $I^{\downarrow}[Y]$, and  say that $B$ (resp.\ $Y$) is {\em Galois-stable} if $B = B^{\uparrow\downarrow}$ (resp.\ $Y = Y^{\downarrow\uparrow}$). When $B=\{a\}$ (resp.\ $Y=\{x\}$) we write $a^{\uparrow\downarrow}$ for $\{a\}^{\uparrow\downarrow}$ (resp.~$x^{\downarrow\uparrow}$ for $\{x\}^{\downarrow\uparrow}$). Galois-stable sets are the projections  of some maximal rectangle (formal concept) of $\mathbb{P}$.
The following lemma collects more well known facts (cf.~\cite[Sections 7.22-7.29]{DaveyPriestley2002}):
\begin{lemma}
\label{lemma: basic stable}
\begin{enumerate}
\item $B^{\uparrow}$ and $Y^{\downarrow}$ are Galois-stable.
\item $B=\bigcup_{a\in B}a^{\uparrow\downarrow}$ and $Y=\bigcup_{y\in Y}y^{\downarrow\uparrow}$ for any Galois-stable $B$ and $Y$.
\item Galois-stable sets are closed under arbitrary intersections.
\end{enumerate}
\end{lemma}
\begin{proof} For item 2, since $a^{\uparrow\downarrow}\supseteq\{a\}$, we have that $B\subseteq\bigcup_{a\in B}a^{\uparrow\downarrow}$. For the other direction, if $\{a\}\subseteq B$ then $a^{\uparrow\downarrow}\subseteq B^{\uparrow\downarrow}$. Since $B$ is Galois-stable, we have that $B=B^{\uparrow\downarrow}$. Hence $a^{\uparrow\downarrow}\subseteq B$ for any $a\in B$, which implies that $\bigcup_{a\in B}a^{\uparrow\downarrow}\subseteq B$. The proof for $Y$ is analogous.
\end{proof}
\begin{definition}
\label{def:I-compatible rel}
	For any  $\mathbb{P}=(A,X,I)$, any $R\subseteq A\times X$  is $I$-{\em compatible} if $R^{\downarrow}[x]$ 
and $R^{\uparrow}[a]$  
are Galois-stable for all $x$ and $a$.
\end{definition}
By Lemma \ref{lemma: basic} (3), $I$ is an $I$-compatible relation.
\begin{lemma}\label{lemma:comp1}
	If $R\subseteq A\times X$ is  $I$-compatible,  then $R^\downarrow[Y]=R^\downarrow[Y^{\downarrow\uparrow}]$ and $R^\uparrow[B]=R^\uparrow[B^{\uparrow\downarrow}]$.
\end{lemma}
\begin{proof} By Lemma \ref{lemma: basic} (2), we have $Y\subseteq Y^{\downarrow\uparrow}$, which implies  $R^\downarrow[Y^{\downarrow\uparrow}]\subseteq R^\downarrow[Y]$ by Lemma \ref{lemma: basic} (1). Conversely, if $a\in R^\downarrow[Y]$, i.e.\ $Y\subseteq R^{\uparrow}[a]$, then $Y^{\downarrow\uparrow}\subseteq (R^{\uparrow}[a])^{\downarrow\uparrow} = R^{\uparrow}[a]$, the last identity holding since $R$ is $I$-compatible. Hence, $a\in R^\downarrow[Y^{\downarrow\uparrow}]$, as required. The proof of the second identity is similar.
%
%
\end{proof}

\begin{lemma}\label{lemma:comp2}
	If $R$ is  $I$-compatible and $Y$ is Galois-stable, then $R^{\downarrow}[Y]$ is Galois-stable.
\end{lemma}
\begin{proof}
  Since $Y=\bigcup_{y\in Y}\{y\}$, by Lemma \ref{lemma: basic} (4),
	\begin{align}
	R^{\downarrow}[Y]=R^{\downarrow}[\bigcup_{y\in Y}\{y\}]=\bigcap_{y\in Y}R^{\downarrow}[\{y\}]=\bigcap_{y\in Y}R^{\downarrow}[y].
\end{align}
By the $I$-compatibility of $R$, the last term is an intersection of Galois-stable sets, which is Galois-stable (cf.~Lemma \ref{lemma: basic stable} (3)).
\end{proof}
The lemma above ensures that  the interpretation of $\mathcal{L}$-formulas   on enriched formal contexts defines a compositional semantics on formal concepts if the relations $R_i$ are  $I$-compatible.
Indeed, for every enriched formal context $\mathbb{F} = (\mathbb{P}, \{R_i\mid i\in \Ag\})$, every valuation $V$ on $\mathbb{F}$ extends to an interpretation map of $\mathcal{L}$-formulas defined as follows:
\begin{center}
\begin{tabular}{r c l c r c l}
$V(p)$ & $ = $ & $(\val{p}, \descr{p})$
& \quad\quad\quad\quad
& $V(\phi\wedge\psi)$ & $ = $ & $(\val{\phi}\cap \val{\psi}, (\val{\phi}\cap \val{\psi})^{\uparrow})$
\\
$V(\top)$ & $ = $ & $(A, A^{\uparrow})$
&
& $V(\phi\vee\psi)$ & $ = $ & $((\descr{\phi}\cap \descr{\psi})^{\downarrow}, \descr{\phi}\cap \descr{\psi})$ \\
 $V(\bot)$ & $ = $ & $(X^{\downarrow}, X)$
&
& $V(\Box_i\phi)$ & $ = $ & $(R_i^{\downarrow}[\descr{\phi}], (R_i^{\downarrow}[\descr{\phi}])^{\uparrow})$\\
\end{tabular}
\end{center}
By Lemma \ref{lemma:comp2}, if $V(\phi)$ is a formal concept, then so is $V(\Box_i\phi)$.
\begin{definition}
\label{def:compositional efc}
An enriched formal context $\mathbb{F} = (\mathbb{P}, \{R_i\mid i\in \Ag\})$ is {\em compositional} if $R_i$ is $I$-compatible (cf.~Definition \ref{def:I-compatible rel}) for every $i\in \Ag$. A model $\mathbb{M} = (\mathbb{F}, V)$ is {\em compositional} if so is $\mathbb{F}$.
\end{definition}
\subsection{The interpretation of $C$ is well defined}
\label{ssec:semantics common knowledge}
 For any formal context $\mathbb{P}=(A,X,I)$ the $I$-{\em product} of the relations $R_s, R_t\subseteq A\times X$ is the relation  $R_{st}\subseteq A\times X$ defined as follows: \[a\in R_{st}^{\downarrow}[x]\mbox{ iff }  a\in R^\downarrow_s\left[I^\uparrow\left[R^\downarrow_t[x^{\downarrow\uparrow}]\right]\right].\]

\begin{lemma}\label{lemma:comp3}
	If $R_s$ and $R_t$ are $I$-compatible, then $R_{st}$ is $I$-compatible.
\end{lemma}
\begin{proof}
$R^{\downarrow}_{st}[x]$ being Galois-stable follows from the definition of $R_{st}$, Lemma \ref{lemma:comp2}, and  the $I$-compatibility of $R_s$ and $R_t$.
	To show that $R_{st}^{\uparrow}[a]$ is Galois-stable, i.e.~$(R_{st}^{\uparrow}[a])^{\downarrow\uparrow}\subseteq R_{st}^{\uparrow}[a]$, by Lemma \ref{lemma: basic stable} (2), it is enough to show that  if $y\in R_{st}^{\uparrow}[a]$ then $y^{\downarrow\uparrow}\subseteq R_{st}^{\uparrow}[a]$. Let $y\in R_{st}^{\uparrow}[a]$, i.e.~$a\in R_{st}^{\downarrow}[y] = R^\downarrow_s\left[I^\uparrow\left[R^\downarrow_t[y^{\downarrow\uparrow}]\right]\right]$. If
$x\in y^{\downarrow\uparrow}$, then $x^{\downarrow\uparrow}\subseteq y^{\downarrow\uparrow}$, which implies, by the antitonicity of  $R_s^\downarrow$, $I^\uparrow$ and $R^\downarrow_t$ (cf.~Lemma \ref{lemma: basic} (1)), that $R^\downarrow_s\left[I^\uparrow\left[R^\downarrow_t[y^{\downarrow\uparrow}]\right]\right]\subseteq R^\downarrow_s\left[I^\uparrow\left[R^\downarrow_t[x^{\downarrow\uparrow}]\right]\right]$. Hence, $a\in R_{st}^{\downarrow}[x]$, i.e.~$x\in R_{st}^{\uparrow}[a]$, as required.
\end{proof}
The definition of $I$-product serves to characterize semantically  the  relation associated with the modal operators  $\Box_s:=\Box_{i_1}\cdots\Box_{i_n}$ for any finite nonempty sequence $s: = i_1\cdots i_n\in S$ of elements of $\Ag$, in terms of the relations associated with each primitive modal operator. For any such $s$, let $R_s$ be defined recursively as follows:
\begin{itemize}
	\item If $s=i$, then $R_s=R_i$;
	\item If $s=it$, then $R_s^{\downarrow}[x] =  R^\downarrow_i\left[I^\uparrow\left[R^\downarrow_t\left[x^{\downarrow\uparrow}\right]\right]\right]$.
\end{itemize}

Lemma \ref{lemma:comp3} immediately implies that
\begin{corollary}
 For every $s\in S$, the relation $R_s$ is $I$-compatible.
\end{corollary}

\begin{lemma}\label{lemma:comp4}
If $Y$ is	Galois-stable  and $R_s,R_t$ are $I$-compatible, then $R^\downarrow_{st}[Y]=R^\downarrow_s[I^\uparrow[R^\downarrow_t[Y]]]$.
\end{lemma}
\begin{proof}$\quad$
	\begin{center}
\begin{tabular}{r c l l}
	    $R_s^\downarrow[I^\uparrow[R^\downarrow_t[Y]]]$ &
     = & $R_s^\downarrow[I^\uparrow[R^\downarrow_t[\bigcup_{x\in Y}x^{\downarrow\uparrow}]]]$ & Lemma \ref{lemma: basic stable} (2)\\
	& = &  $R_s^\downarrow[I^\uparrow[\bigcap_{x\in Y}R^\downarrow_t[x^{\downarrow\uparrow}]]]$ & Lemma \ref{lemma: basic} (4)  \\
    & = &  $R_s^\downarrow[I^\uparrow[\bigcap_{x\in Y}I^\downarrow[I^\uparrow[R^\downarrow_t[x^{\downarrow\uparrow}]]]]]$ & $R^\downarrow_t[x^{\downarrow\uparrow}]$ Galois-stable  \\

	& = &  $R_s^\downarrow[I^\uparrow[I^\downarrow[\bigcup_{x\in Y}I^\uparrow[R^\downarrow_t[x^{\downarrow\uparrow}]]]]]$ & Lemma \ref{lemma: basic} (4)\\
	& = &  $R_s^\downarrow[\bigcup_{x\in Y}I^\uparrow[R^\downarrow_t[x^{\downarrow\uparrow}]]] $ & Lemma \ref{lemma:comp1}\\
	& = & $\bigcap_{x\in Y} R_s^\downarrow[I^\uparrow[R^\downarrow_t[x^{\downarrow\uparrow}]]]$ & Lemma \ref{lemma: basic} (4)\\
	& = & $\bigcap_{x\in Y} R_{st}^\downarrow[x]$ & Definition of $R_{st}$ \\
	& = & $ R_{st}^\downarrow[\bigcup_{x\in Y}x]$ & Lemma \ref{lemma: basic} (4) \\
& = & $ R_{st}^\downarrow[Y]$ & $Y = \bigcup_{x\in Y}x$ \\
	\end{tabular}
\end{center}
\end{proof}

\begin{lemma}
\label{lemma:associativity of I product}
	If $R_s,R_t,R_w$ are $I$-compatible,  $R_{s(tw)}=R_{(st)w}$.
\end{lemma}
\begin{proof} for any $x$,
	\begin{center}
\begin{tabular}{r c l l}
	   $R^{\downarrow}_{s(tw)}[x]$ 
   & = & $R^\downarrow_s[I^\uparrow[R^\downarrow_{tw}[x^{\downarrow\uparrow}]]]$ & definition of $I$-product\\
   & = & $R^\downarrow_s[I^\uparrow[R^\downarrow_{t}[I^\uparrow[R^\downarrow_{w}[x^{\downarrow\uparrow}]]]]]$ & Lemma \ref{lemma:comp4}\\
& = & $R^\downarrow_{st}[I^\uparrow[R^\downarrow_{w}[x^{\downarrow\uparrow}]]]$. & Lemma \ref{lemma:comp4}\\
 &= &   $R^{\downarrow}_{(st)w}[x]$ & definition of $I$-product \\
 \end{tabular}
 \end{center}
\end{proof}

Let $s=i_1\cdots i_n\in S$, and let $\Box_s:=\Box_{i_1}\cdots\Box_{i_n}$.
\begin{lemma}
For any model $\mathbb{M}=(\mathbb{F},V)$,
	\begin{center}
		\begin{tabular}{llll}
			$\mathbb{M}, a \Vdash \Box_s\phi$ & iff & for all $x\in X$, if $\mathbb{M}, x \succ \phi$, then $a R_s x$& \\
			$\mathbb{M}, x \succ \Box_s\phi$ & iff & for all $a\in A$, if $\mathbb{M}, a \Vdash \Box_s\phi$, then $a I x$.\\
			%
		\end{tabular}
	\end{center}
\end{lemma}
\begin{proof}
By induction on the length of $s$. The base case is immediate. Let $s=it$. Then
$\val{\Box_i\Box_t\phi} = R^\downarrow_i[\descr{\Box_t\phi}] = R^\downarrow_i[I^\uparrow[\val{\Box_t\phi}]] = R^\downarrow_i[I^\uparrow[R^\downarrow_t[\descr{\phi}]]] = R^\downarrow_s[\descr{\phi}].$
The last equality holds by Lemma \ref{lemma:comp4}.
The second equivalence is trivially true.
\end{proof}

\begin{lemma}
\label{lemma:bigcap I compatible rels }
	For any family $\mathcal{R}$ of $I$-compatible relations,
\begin{enumerate}
\item $\bigcap\mathcal{R}$  is an $I$-compatible relation.
\item $(\bigcap\mathcal{R})^{\downarrow}[Y] = \bigcap_{T\in \mathcal{R}}T^{\downarrow}[Y]$ for any   $Y\subseteq X$.
\end{enumerate}
\end{lemma}
\begin{proof}
 Let  $R=\bigcap\mathcal{R}$. Then $R^{\downarrow}[x]=\bigcap_{T\in \mathcal{R}}T^{\downarrow}[x]$ and $R^{\uparrow}[a]=\bigcap_{T\in \mathcal{R}}T^{\uparrow}[a]$. Then the statement follows from Lemma \ref{lemma: basic stable} (3).
As to item (2),  
	\begin{center}
	\begin{tabular}{rcll}
		   $\bigcap_{T\in \mathcal{R}}T^{\downarrow}[Y]$ 
		& = &$\bigcap_{T\in \mathcal{R}}T^{\downarrow}[\bigcup_{y\in Y}y]$ & $Y = \bigcup_{y\in Y}y$ \\
		& = &$\bigcap_{T\in \mathcal{R}}\bigcap_{y\in Y}T^{\downarrow}[y]$ & Lemma \ref{lemma: basic} (4)\\
		& = &$\bigcap_{y\in Y}\bigcap_{T\in \mathcal{R}}T^{\downarrow}[y]$ & associativity, commutativity of $\bigcap$\\
		& = &$\bigcap_{y\in Y}(\bigcap\mathcal{R})^{\downarrow}[y]$ & definition of $(\cdot)^\downarrow$\\
		& = &$(\bigcap\mathcal{R})^{\downarrow}[\bigcup_{y\in Y}y]$ & Lemma \ref{lemma: basic} (4)\\
		& = &$(\bigcap\mathcal{R})^{\downarrow}[Y]. $ & $Y = \bigcup_{y\in Y}y$ %
		%
	\end{tabular}
\end{center}

\end{proof}
The lemmas above ensure that, in enriched formal contexts in which the relations $R_i$ are  $I$-compatible, the relation $R_C: = \bigcap_{s\in S}R_s$ is  $I$-compatible, and hence the interpretation of $\Lmu$-formulas   on the model based on these enriched formal contexts defines a compositional semantics on formal concepts.
Indeed, for every such enriched formal context $\mathbb{F} = (\mathbb{P}, \{R_i\mid i\in \Ag\})$, every valuation $V$ on $\mathbb{F}$ extends to an interpretation map of $C$-formulas  as follows:
\begin{center}
\begin{tabular}{r c l}
$V(C(\phi))$ & $ = $ & $(R_C^{\downarrow}[\descr{\phi}], (R_C^{\downarrow}[\descr{\phi}])^{\uparrow})$\\
\end{tabular}
\end{center}
so that if $V(\phi)$ is a formal concept, then so is $V(\Box_i\phi)$. Moreover, the following identity is semantically supported:
\[C(\phi) = \bigwedge_{s\in S}\Box_s\phi,\]
where $s: = i_1\cdots i_n$ is any finite nonempty string of elements of $\Ag$, and $\Box_s: = \Box_{i_1}\cdots \Box_{i_n}$.
\subsection{Soundness}
\label{ssec:soundness}
\begin{proposition}
\label{prop:soundness monotonicity multipl}
For any compositional model $\mathbb{M}$ and any $i\in \Ag$,
\begin{enumerate}
\item if $\mathbb{M}\models \phi\vdash \psi$, then $\mathbb{M}\models \Box_i\phi\vdash \Box_i\psi$;
\item  $\mathbb{M}\models \top\vdash \Box_i\top$;
\item $\mathbb{M}\models \Box_i\phi\wedge \Box_i\psi\vdash \Box_i(\phi\wedge\psi)$.
\end{enumerate}	
\end{proposition}
\begin{proof}
By Lemma \ref{lemma: basic} (1), if $\val{\phi}\subseteq \val{\psi}$ then $$\val{\Box_i\phi} = R_i^\downarrow[I^\uparrow[\val{\phi}]]\subseteq R_i^\downarrow[I^\uparrow[\val{\psi}]] = \val{\Box_i\psi},$$ which proves item (1). 
As to item (2), it is enough to show that $\val{\Box_i\top}=A$. By definition, $\val{\Box_i\top}=R_i^\downarrow[\descr{\top}]=R_i^\downarrow[A^\uparrow],$ hence it is enough to show that $R_i^\downarrow[A^\uparrow]=A$. The assumption of $I$-compatibility implies that $R_i^\uparrow[a]$ is Galois-stable for every $a\in A$, and hence $A^\uparrow\subseteq R_i^\uparrow[a]$. Thus by adjunction $a\in R_i^\downarrow[A^\uparrow]$ for every $a\in A$, which implies that $R_i^\downarrow[A^\uparrow]=A$, as required. 
As to item (3),
	\begin{center}
		\begin{tabular}{rlll}
			$\val{\Box(\phi)\land \Box(\psi)}$ 
            & = &$R^\downarrow[\descr{\phi}]\cap R^\downarrow[\descr{\psi}] $& definition of $\val{\cdot}$ \\
			& = &$R^\downarrow[\descr{\phi}\cup\descr{\psi}] $& Lemma \ref{lemma: basic} (4)\\
			& = &$R^\downarrow[I^\uparrow[I^\downarrow[\descr{\phi}\cup\descr{\psi}]]]$& Lemma \ref{lemma:comp1}\\
			& = &$R^\downarrow[I^\uparrow[I^\downarrow[\descr{\phi}]\cap I^\downarrow[\descr{\psi}]]]$& Lemma \ref{lemma: basic} (4)\\
			& = &$R^\downarrow[I^\uparrow[\val{\phi}\cap \val{\psi}]]$& $V(\phi), V(\phi)$ formal concepts\\
			& = &$ R^\downarrow[I^\uparrow[\val{\phi\land \psi}]]$& definition of $\val{\cdot}$ \\
			& = & $\val{\Box(\phi\land \psi)}$.& definition of $\val{\cdot}$ \\
			%
			%
		\end{tabular}
	\end{center}
\end{proof}
\begin{proposition}
\label{prop:soundness}
For any compositional model $\mathbb{M}$,
\begin{enumerate}
\item $\mathbb{M}\models C(\phi)\vdash \bigwedge\{\Box_i\phi\land\Box_i C(\phi)\mid i\in \Ag\}$;
\item if $\mathbb{M}\models \chi\vdash\bigwedge_{i\in \Ag}\Box_i \phi$ and $\mathbb{M}\models \chi\vdash\bigwedge_{i\in \Ag}\Box_i\chi$, then $\mathbb{M}\models \chi\vdash C(\phi)$.
\end{enumerate}	
\end{proposition}
\begin{proof}
By definition and Lemma \ref{lemma:bigcap I compatible rels } (2), $\val{C(\phi)} = R^\downarrow_C[\descr{\phi}]=\bigcap_{s\in S}R^\downarrow_s[\descr{\phi}]\subseteq \bigcap_{i\in \Ag}R^\downarrow_i[\descr{\phi}]$, which proves $\mathbb{M}\models C(\phi)\vdash \bigwedge\{\Box_i\phi\mid i\in \Ag\}$. Let $i\in \Ag$. The following chain of (in)equalities completes the proof of item (1):
	\begin{center}
		\begin{tabular}{rllll}
			 $\val{\Box_iC(\phi)}$ 
            & = & $R^\downarrow_i[I^\uparrow[R^\downarrow_C[\descr{\phi}]]]$& definition of $\val{\cdot}$ \\
			& = & $R^\downarrow_i[I^\uparrow[\bigcap_{s\in S}R^\downarrow_s[\descr{\phi}]]]$& Lemma \ref{lemma:bigcap I compatible rels } (2)\\
            & = & $R^\downarrow_i[I^\uparrow[\bigcap_{s\in S}I^\downarrow[I^\uparrow[R^\downarrow_s[\descr{\phi}]]]]]$& $R^\downarrow_s[\descr{\phi}]$ Galois-stable\\
			& = & $R^\downarrow_i[I^\uparrow[I^\downarrow[\bigcup_{s\in S}I^\uparrow[R^\downarrow_s[\descr{\phi}]]]]]$& Lemma \ref{lemma: basic} (4) \\
			& = & $R^\downarrow_i[\bigcup_{s\in S}I^\uparrow[R^\downarrow_s[\descr{\phi}]]]$& Lemma \ref{lemma:comp1} \\
			& = & $\bigcap_{s\in S}R^\downarrow_i[I^\uparrow[R^\downarrow_s[\descr{\phi}]]]$& Lemma \ref{lemma: basic} (4)\\
            & = & $\bigcap_{s\in S}R^\downarrow_{is}[\descr{\phi}]$& Lemma \ref{lemma:comp4} \\
            & $\supseteq$ & $\bigcap_{s\in S}R^\downarrow_{s}[\descr{\phi}]$& $\{is\mid s\in S\}\subseteq S$\\
            &  =  & $\val{C(\phi)}$.& Lemma \ref{lemma:bigcap I compatible rels } (2)\\
			%
			%
		\end{tabular}
	\end{center}
As to item (2), using Proposition \ref{prop:soundness monotonicity multipl} (1) and the assumptions, one can show that $\mathbb{M}\models \chi\vdash \Box_s\phi$ for every $s\in S$. Hence, $\val{\chi}\subseteq \bigcap_{s\in S}R_{s}^{\downarrow}[\descr{\phi}] = R_C^{\downarrow}[\descr{\phi}] = \val{C(\phi)}$, as required.
%
\end{proof}
\subsection{Completeness}
\label{ssec:completeness}
The completeness of $\mathbf{L}$ can be proven via a standard canonical model construction. For any lattice $\mathbb{L}$ with normal operators $\Box_i$, let $\mathbb{F}_{\mathbb{L}} = (\mathbb{P}_{\mathbb{L}}, \{R_i\mid i\in \Ag\})$ be defined as follows: $\mathbb{P}_{\mathbb{L}} = (A, X, I)$ where $A$ (resp.~$X$) is the set of lattice filters (resp.~ideals) of $\mathbb{L}$, and $aIx$ iff $a\cap x \neq \varnothing$. 
For every $i\in \Ag$, let $R_{i}\subseteq A\times X$ be defined by $aR_ix$ iff   if $\Box_i u\in a$  for some $u\in \mathbb{L}$ such that $u\in x$. 
In what follows, for any $a\in A$ and $x\in X$,  we let $\Box_ix: =\{\Box_i u\in\mathbb{L}\mid u\in x\}$ and $\Box_i^{-1}a:=\{u\in\mathbb{L}\mid \Box_i u\in a\}$. Hence by definition, $R_i^{\downarrow}[x] = \{a\mid a\cap \Box_i x\neq \varnothing\}$ for any $x\in X$, and $R_i^{\uparrow}[a] = \{x\mid x\cap \Box_i^{-1}a\neq \varnothing \}$ for any $a\in A$. Notice also that $\Box_i\top=\top$ implies that $\Box_{i}^{-1}a=\neq\varnothing$ for every $a\in A$.
\begin{lemma}
\label{lemma:prelim fl}
For $\mathbb{F}_{\mathbb{L}}$ as above, and any $a\in A$, $x\in X$ and $i\in \Ag$,
\begin{enumerate}
\item $I^\uparrow[R_i^\downarrow[x]]=\{y\in X\mid \Box_ix\subseteq y\}$;
\item $I^\downarrow[R_i[a]]=\{b\in A\mid \Box_i^{-1}a\subseteq b\}$;
\item $I^{\downarrow}[I^\uparrow[R_i^\downarrow[x]]]= \{b\in A\mid \Box_ix \cap b\neq\varnothing \} = R_i^\downarrow[x]$;
\item $I^{\uparrow}[I^\downarrow[R_i^{\uparrow}[a]]]=\{y\in X\mid \Box_i^{-1}a\cap y\neq \varnothing\} = R_i^{\uparrow}[a]$.
\end{enumerate}
\end{lemma}
\begin{proof}
Items (1) and (2) immediately follow from the definitions of $\Box_ix$ and $\Box_i^{-1}a$. As to items (3) and (4), from the previous items it immediately follows that $I^{\downarrow}[I^\uparrow[R_i^\downarrow[x]]]= \{b\in A\mid \lceil\Box_ix\rceil \cap b\neq\varnothing \}$ and $I^{\uparrow}[I^\downarrow[R_i^{\uparrow}[a]]]=\{y\in X\mid \lfloor\Box_i^{-1}a\rfloor\cap y\neq \varnothing\}$, where $\lceil\Box_ix\rceil$ and $\lfloor\Box_i^{-1}a\rfloor$ respectively denote the ideal generated $\Box_ix$ and the filter generated by $\Box_i^{-1}a$. Then, using the monotonicity of $\Box_{i}$, one can show that   $\{b\in A\mid \lceil\Box_ix\rceil \cap b\neq\varnothing \}=\{b\in A\mid \Box_ix \cap b\neq\varnothing \}=R_i^\downarrow[x]$, and using the meet preservation of  $\Box_{i}$, one can show that $\{y\in X\mid \lfloor\Box_i^{-1}a\rfloor\cap y\neq \varnothing\}=\{y\in X\mid \Box_i^{-1}a\cap y\neq \varnothing\}=R_i^{\uparrow}[a]$, as required. Notice that the last equality holds for every $a\in A$ under the assumption that $\Box_i^{-1}a\neq \varnothing$, which, as remarked above, is guaranteed by $\Box_i$ being normal. 
\end{proof}
Items (3) and (4) of the lemma above immediately imply that:
\begin{lemma}\label{lemmaforcorrectness}
	 $\mathbb{F}_\mathbb{L}$ is a compositional enriched formal context (cf.~Definition \ref{def:compositional efc}).
\end{lemma}
Recall that $S$ is the set of nonempty finite sequences of elements of $\Ag.$
\begin{lemma}\label{mainlemmaforeazycompleteness}
	If $x$ is the ideal generated by some $u\in\mathbb{L}$, then, for every $s\in S$, $R_s^{\downarrow}[x] = \{a\mid \Box_su\in a\}$.
\end{lemma}
\begin{proof}
By induction on the length of $s\in S$. If $s=i$ then $aR_i x$ iff $a\in R_i^{\downarrow}[x]$ iff $a\cap\Box_i x\neq\varnothing$. Since $x$ is the ideal generated by $u$, we have that $u$ is the greatest element of $x$; hence, the monotonicity of $\Box_i$ implies that $\Box_iu$ is the greatest element of  $\Box_i x$. Since $a$ is a filter, and hence is upward-closed,  $a\cap\Box_i x\neq\varnothing$ is equivalent to $\Box_iu\in a$, which completes the proof of the base case.
Let us assume that $R^\downarrow_s[x]=\{b\in A\mid \Box_su\in b\}$, and  show that $R^\downarrow_{is}[x]=\{b\in A\mid \Box_{is}u\in b\}$.
By Lemma \ref{lemma:prelim fl} (3) and (4), and Lemma \ref{lemma:comp3}, $R_s$ is $I$-compatible for every $s\in S$. Let $z$ be the ideal generated by $\Box_s u$. Hence:
\begin{center}
\begin{tabular}{ r c ll}
 $R^\downarrow_{is}[x]$ 
& = & $R^\downarrow_{i} [I^\uparrow[R^\downarrow_s[x]]]$ & Lemmas \ref{lemma:comp1} and  \ref{lemma:comp4}\\
& = & $R^\downarrow_{i} [(\{b\in A\mid \Box_su\in b\})^\uparrow]$ & induction hypothesis\\
& = & $R^\downarrow_{i} [\{y\in X\mid \Box_s u\in y\}]$ & $(\ast)$\\
&  = & $R^\downarrow_{i} [z]$ & definition of $z$ \\
&   = & $\{a\mid \Box_i\Box_su\in a\}$ & base case\\
&    = & $\{a\mid \Box_{is}u\in a\}$. & definition of $\Box_{is}$\\
\end{tabular}
\end{center}
The identity marked with $(\ast)$ follows from the fact that the filter generated by $\Box_su$ is the smallest element of  $R^\downarrow_s[x]$.
\end{proof}

The {\em canonical enriched formal context} is defined by instantiating the construction above to the Lindembaum-Tarski algebra of $\mathbf{L}$. In this case, let $V$ be the valuation such that $\val{p}$ (resp.\ \descr{p}) is the set of the  filters (resp.\ ideals) to which $p$ belongs, and let $\mathbb{M} = (\mathbb{F}_{\mathbf{L}}, V)$ be the canonical model. Then the following holds for $\mathbb{M}$:
\begin{lemma}[Truth lemma]
For every $\phi\in \mathcal{L}$,
\begin{enumerate}
\item $\mathbb{M}, a\Vdash \phi$ iff $\phi\in a$;
\item $\mathbb{M}, x\succ \phi$ iff $\phi\in x$.
\end{enumerate}
\end{lemma}
\begin{proof}
By induction on $\phi$. We only show the inductive step for $\phi: = \Box_i\sigma$.
		\begin{center}
	\begin{tabular}{rcll}
	    $\mathbb{M}, a\Vdash \Box_i\sigma$ 
  &  iff & $a\in R^\downarrow_i[\descr{\sigma}]$\\
  &  iff & $a\in R^\downarrow_i[\{x\mid\sigma\in x\}]$ & induction hypothesis\\
&	iff & $a\in \{b\in A\mid \Box_i\sigma\in b\}$ & definition of $R_i$\\
	&iff & $\Box_i\sigma\in a$.\\
	    &&&\\
	     $\mathbb{M}, x\succ \Box_i\sigma$ 
  &  iff & $x\in \descr{\Box_i\sigma}$ \\ 
&	iff & $x\in \val{\Box_i\sigma}^{\uparrow}$\\ 
&	iff & $x\in \{a\in A\mid \Box_i\sigma\in a\}^{\uparrow}$ & proof above\\
&iff & $\Box_i\sigma\in x$.\\
\end{tabular}
\end{center}
\end{proof}
The weak completeness of $\mathbf{L}$ follows from the lemma above with the usual argument.
\begin{proposition}[Completeness]
If $\varphi\vdash\psi$ is an $\mathcal{L}$-sequent which is not derivable in $\mathbf{L}$, then $\mathbb{M}\not\models \varphi\vdash\psi$.
\end{proposition}

\vspace{20pt}

 The  weak completeness for $\mathbf{L}_C$ is proved along the lines of \cite[Theorem 3.3.1]{fagin2003reasoning}.  Namely, for any $\mathcal{L}_C$-sequent $\varphi\vdash\psi$ that is not derivable in $\mathbf{L}_C$, we will construct a finite model $\mathbb{M}_{\varphi,\psi}$  such that $\mathbb{M}_{\varphi,\psi}\not\models\varphi\vdash\psi$. 
Let $\Phi_0$ be the set the elements of which are $\top$, $\bot$ and all the subformulas of $\varphi$ and $\psi$. 
Let 
\[\Phi_1:=\Phi_0\cup\bigcup_{i\in\Ag}\{\Box_i\sigma\mid\sigma\in\Phi_0\} \quad\quad \text{and} \quad\quad \Phi:=\{\bigwedge\Psi\mid\Psi\subseteq\Phi_1\}.\] 
By construction,  $\Phi$ is  finite.
Consider the canonical model $\mathbb{M}$ defined above, and  the following equivalence relations on $A$ and $X$: 
\[a\equiv_\Phi b \mbox{ iff } a\cap\Phi=b\cap\Phi \quad\quad \text{and}\quad\quad  x\equiv_\Phi y \mbox{ iff }x\cap\Phi=y\cap\Phi.\]
Since $\Phi$ is  finite, these equivalence relations induce finitely many equivalence classes on $A$ and $X$. In particular, considering $\vdash$ as a preorder on $\Phi$, each element $\overline{a}$ of $A/{\equiv_\Phi}$ is uniquely identified by some $\Phi$-{\em filter}, i.e.~a $\vdash$-upward closed subset  of $\Phi$ which is also closed under existing conjunctions.   Analogously, each  $\overline{x}\in X/{\equiv_\Phi}$ is uniquely identified by some $\Phi$-{\em ideal}, i.e.~a $\vdash$-downward closed subset  of $\Phi$ which is also closed under existing disjunctions.   In addition, since $\Phi$ is  closed under conjunctions, the $\Phi$-filter corresponding to each $\overline{a}$ is {\em principal}, i.e.~for each $\overline{a}\in A/{\equiv_\Phi}$ some $ \tau_{\overline{a}}\in \Phi$ exists such that $\overline{a}$ can be identified with the set of the formulas $\sigma\in \Phi$ such that  $\tau_{\overline{a}}\vdash \sigma$ is an $\mathbf{L}_C$-derivable sequent. 
In what follows, we abuse notation and let $\overline{a}$ and $\overline{x}$ respectively denote the principal $\Phi$-filter and the $\Phi$-ideal with which $\overline{a}$ and $\overline{x}$ can be identified, as discussed above. With this convention, we can write $\Box^\ast_i\overline{x}:=\{\Box_i\sigma\mid \sigma\in \overline{x}\}\cap\Phi$ and $(\Box_i^{-1})^\ast \overline{a}:=\{\tau\in\Phi\mid \Box_i \tau\in \overline{a}\}$. As a consequence of $\bot,\top\in\Phi_0$ and $\Box_i\top=\top$ we have that $\Box^\ast_i\overline{x}$ and $(\Box_i^{-1})^\ast \overline{a}$ are always non-empty. Let us define: $$\mathbb{M}_{\varphi,\psi}=(A/{\equiv_\Phi}, X/{\equiv_\Phi}, I_{\varphi,\psi}, R^{\varphi,\psi}_i, V_{\varphi,\psi}),$$ where
\begin{center}
\begin{tabular}{r c l c l}
$\overline{a}I_{\varphi,\psi}\overline{x}$ &   iff & $\overline{a}\cap \overline{x}\neq\varnothing$ &   iff & $\tau_{\overline{a}}\in  \overline{x}$ \\
$\overline{a}R^{\varphi,\psi}_i\overline{x}$ &   iff & $\Box^\ast_i \overline{x}\cap \overline{a}\neq\varnothing$ \\ &   iff & \multicolumn{3}{l}{$\tau_{\overline{a}}\vdash \Box_i \tau$ is $\mathbf{L}_C$-derivable for some $\tau\in \overline{x}$,}\\
 \end{tabular}
 \end{center}
 and $V_{\varphi,\psi}$ is any valuation such that  $\val{p}=\{\overline{a}\mid p\in\overline{a}\}$ and  $\descr{p}=\{\overline{x}\mid p\in\overline{x}\}$ for all $p\in\Prop\cap \Phi$. In what follows, we often abbreviate $I_{\varphi,\psi}$ as $I$. It readily follows from the definition that $\val{p}^{\uparrow\downarrow}=\val{p}$ and $\descr{p}^{\downarrow\uparrow}=\descr{p}$ for any $p\in\Prop\cap \Phi$; moreover,
%
$(R^{\varphi,\psi}_i)^{\downarrow}[\overline{x}] = \{\overline{a}\mid \overline{a}\cap \Box^\ast_i\overline{x}\neq \varnothing\}$, and $(R^{\varphi,\psi}_i)^{\uparrow}[\overline{a}] = \{\overline{x}\mid \overline{x}\cap (\Box_i^{-1})^\ast \overline{a}\neq \varnothing \}$. From this, similarly to Lemma \ref{lemma:prelim fl}, it immediately follows that:

\begin{lemma}
	\label{lemma:prelim fl2}
	For  any $\overline{a}$, $\overline{x}$ and $i\in \Ag$,
	\begin{enumerate}
		\item $I_{\varphi,\psi}^\uparrow[(R^{\varphi,\psi}_i)^\downarrow[\overline{x}]]=\{\overline{y}\mid \Box^\ast_i\overline{x}\subseteq \overline{y}\}$;
		\item $I_{\varphi,\psi}^\downarrow[(R^{\varphi,\psi}_i)^\uparrow[\overline{a}]]=\{\overline{b}\in A\mid (\Box_i^{-1})^\ast \overline{a}\subseteq\overline{b}\}$;
		\item $I_{\varphi,\psi}^{\downarrow}[I_{\varphi,\psi}^\uparrow[(R^{\varphi,\psi}_i)^\downarrow[\overline{x}]]]= \{\overline{b}\mid \Box^\ast_i\overline{x}\cap\overline{b}\neq\varnothing \} = (R^{\varphi,\psi}_i)^\downarrow[\overline{x}]$;
		\item $I_{\varphi,\psi}^{\uparrow}[I_{\varphi,\psi}^\downarrow[(R^{\varphi,\psi}_i)^{\uparrow}[\overline{a}]]]=\{\overline{y}\mid (\Box_i^{-1})^\ast \overline{a}\cap \overline{y}\neq \varnothing\} = R_i^{\uparrow}[\overline{a}]$.
	\end{enumerate}
\end{lemma}
Items (3) and (4) of the lemma above immediately imply that:

\begin{lemma}
\label{lemma:M phi psi is compositional}
 $R^{\varphi,\psi}_i$ is $I_{\varphi,\psi}$-compatible for any $i\in \Ag$.
\end{lemma}

The following is key to the proof of the Truth Lemma.

\begin{lemma}\label{lemma:keyforfixedpointcomp}
	If $C(\sigma)\in\Phi$, then the following is an $\mathbf{L}_C$-derivable sequent for any $i\in\Ag$: $$\bigvee_{\overline{a}\in\val{C(\sigma)}}\tau_{\overline{a}} \ \vdash\ \Box_i(\bigvee_{\overline{a}\in \val{C(\sigma)}}\tau_{\overline{a}}).$$
\end{lemma}
\begin{proof}
	Fix  $i\in\Ag$ and $\overline{a}\in\val{C(\sigma)}$. Since $\Box_i$ is monotone, it is enough to show that  some $\tau\in \Phi$ exists such that 
	\[\tau_{\overline{a}}\ \vdash \ \Box_i\tau \quad \mbox{ and } \quad \tau\ \vdash \ \bigvee_{\overline{a}\in\val{C(\sigma)}}\tau_{\overline{a}}.\]
By definition of $R^{\varphi,\psi}_i$, this is equivalent to showing that $\overline{a}R^{\varphi,\psi}_i \overline{y}$, where $\overline{y}$ is the $\Phi$-ideal generated by $\bigvee_{\overline{a}\in\val{C(\sigma)}}\tau_{\overline{a}}$. Notice that $\descr{C(\sigma)} = \val{C(\sigma)}^\uparrow$ is the collection of all the $\Phi$-ideals $\overline{x}$ such that $\tau_{\overline{b}}\in \overline{x}$ for every $\overline{b}\in\val{C(\sigma)}$. Hence, $\overline{y}\in \descr{C(\sigma)}$ (and is in fact the smallest element in $\descr{C(\sigma)}$). Thus, to  prove that $\overline{a}R^{\varphi,\psi}_i \overline{y}$, it is enough to show that $\val{C(\sigma)}\subseteq (R^{\varphi,\psi}_s)^\downarrow[\descr{C(\sigma)}]$. This immediately follows from the fact that $(R^{\varphi,\psi}_s)^\downarrow[\descr{C(\sigma)}] = \val{\Box_iC(\sigma)}$, that $C(\sigma)\vdash \Box_iC(\sigma)$ is an $\mathbf{L}_C$-derivable sequent, that  $\mathbf{L}_C$ is sound w.r.t.~compositional models (cf.~Proposition \ref{prop:soundness}), and $\mathbb{M}_{\varphi,\psi}$ is a compositional model (cf.~Lemma \ref{lemma:M phi psi is compositional}).
\end{proof}

\begin{lemma}[Truth lemma]
	For every $\tau\in\Phi_0$,
	\begin{enumerate}
		\item $M_{\varphi,\psi}, \overline{a}\Vdash \tau$ iff $\tau\in \overline{a}$;
		\item $M_{\varphi,\psi}, \overline{x}\succ \tau$ iff $\tau\in \overline{x}$.
	\end{enumerate}
\end{lemma}
\begin{proof}
	We only show the inductive step for $\tau: = C(\sigma)$ for some $\sigma\in \Phi_0$. If  $M_{\varphi,\psi}, \overline{a}\Vdash C(\sigma)$, i.e.~$\overline{a}\in \val{C(\sigma)} = \bigcap_{s\in S}(R^{\varphi,\psi}_s)^\downarrow[\descr{\sigma}]$, then $\overline{a}\in(R^{\varphi,\psi}_i)^\downarrow[\descr{\sigma}] = \val{\Box_i\sigma}$ for any $i\in \Ag$. 
By definition,  $\sigma\in \Phi_0$ implies that  $\Box_i\sigma\in\Phi$. Moreover:
 \begin{center}
	\begin{tabular}{rcll}
	      $\overline{a}\in\val{\Box_i\sigma}$ 
  &  iff & $\overline{a}\in (R^{\varphi,\psi}_i)^\downarrow[\descr{\sigma}]$\\
  &  iff & $\overline{a}\in (R^{\varphi,\psi}_i)^\downarrow[\{\overline{x}\mid\sigma\in \overline{x}\}]$ & induction hypothesis\\
&	iff & $\overline{a}\in \{\overline{b}\mid \Box_i\sigma\in \overline{b}\}$ & definition of $R^{\varphi,\psi}_i$\\
&	iff & $\Box_i\sigma\in \overline{a}$.\\
\end{tabular}
\end{center}
 This implies that  $\tau_{\overline{a}}\vdash\bigwedge_{i\in \Ag}\Box_i\sigma$.  By Lemma \ref{lemma:keyforfixedpointcomp} and the fact that $\mathbf{L}_{C}$ is closed under the following rule:
		\begin{displaymath}
 \frac{\chi\vdash \bigwedge_{i\in \Ag}\Box_i\varphi\quad \{\chi\vdash \Box_i\chi\mid i\in \Ag\}}{\chi\vdash C\left(\varphi\right)}
	\end{displaymath}
we conclude that  $\tau_{\overline{a}}\vdash C(\sigma)$, i.e.~$C(\sigma)\in\overline{a}$.

For the converse direction,  let $\overline{b}$ be the principal $\Phi$-filter generated by $C(\sigma)$. 
Let us show, by induction on the length of $s$, that $\overline{b}\in(R^{\varphi,\psi}_s)^\downarrow[\descr{\sigma}]$ for all $s\in S$. Indeed, for the base case, $\Box_i\sigma\in\Phi$ and  $C(\sigma)\vdash\Box_i\sigma$ being an $\mathbf{L}_C$-derivable sequent imply that $\Box_i\sigma\in\overline{b}$, which implies that $\overline{b}\in (R^{\varphi,\psi}_i)^\downarrow[\descr{\sigma}]$. For the inductive step, assume that $\overline{b}\in (R^{\varphi,\psi}_s)^\downarrow[\descr{\sigma}]$. Then every element of $I^\uparrow[(R^{\varphi,\psi}_s)^\downarrow[\descr{\sigma}]]$ contains $C(\sigma)$. Moreover, $\Box_iC(\sigma)\in\overline{b}$, because $\Box_iC(\sigma)\in\Phi$ and  $C(\sigma)\vdash\Box_i C(\sigma)$ is an $\mathbf{L}_C$-derivable sequent. Hence, by  Lemma \ref{lemma:comp4}, \[\overline{b}\in (R^{\varphi,\psi}_i)^\downarrow[I^\uparrow[(R^{\varphi,\psi}_s)^\downarrow[\descr{\sigma}]]] = (R^{\varphi,\psi}_{is})^\downarrow[\descr{\sigma}],\] which concludes the proof that $\overline{b}\in(R^{\varphi,\psi}_s)^\downarrow[\descr{\sigma}]$ for all $s\in S$. To finish the proof, for any $\overline{a}$, if $C(\sigma)\in\overline{a}$, then $\overline{b}\subseteq\overline{a}$, which implies, since $(R^{\varphi,\psi}_s)^\downarrow[\descr{\sigma}]$ is Galois-stable for any $s\in S$, that $\overline{a}\in (R^{\varphi,\psi}_s)^\downarrow[\descr{\sigma}]$ for every $s\in S$. This shows that $M_{\varphi,\psi}, \overline{a}\Vdash C(\sigma)$. As to item (2),
\begin{center}
	\begin{tabular}{rcl}
		$M_{\varphi,\psi}, \overline{x}\succ C(\sigma)$ & iff & $\overline{x}\in\val{C(\sigma)}^\uparrow$\\
		& iff & $\overline{x}\cap\overline{a}\neq\varnothing$ for all $\overline{a}\in\val{C(\sigma)}$\\
		& iff & $C(\sigma)\in\overline{x}$.\\
	\end{tabular}
\end{center}
\end{proof}

The weak completeness of $\mathbf{L}_C$ follows from the lemma above with the usual argument.
\begin{proposition}[Completeness]
If $\varphi\vdash\psi$ is an $\mathcal{L}_C$-sequent which is not derivable in $\mathbf{L}_C$, then $M_{\varphi,\psi}\not\models \varphi\vdash\psi$.
\end{proposition}

\end{document}